\documentclass[11pt,a4paper]{article}

\usepackage{enumerate}
\usepackage{eucal}

\usepackage{fullpage}

\usepackage{amsmath,amstext,amssymb,amsthm}

\newcommand{\bref}[1]{(B\ref{#1})}

\newcommand{\ignore}[1]{}

\newcommand{\defparproblemu}[4]{
  \vspace{1mm}
\noindent\fbox{
  \begin{minipage}{0.97\textwidth}
  \begin{tabular*}{0.97\textwidth}{@{\extracolsep{\fill}}lr} #1 & {\bf{Parameter:}} #3 \\ \end{tabular*}
  {\bf{Input:}} #2  \\
  {\bf{Question:}} #4
  \end{minipage}
  }
  \vspace{1mm}
}

\newcommand{\nbname}{{\sc{Nonblocker}}}

\newcommand{\ch}{{\rm ch}}
\newcommand{\sch}{\widehat{\ch}}
\renewcommand{\ae}{\AE{}}

\newcommand{\case}[1]{\noindent{\bf CASE #1}}
\renewcommand{\P}{\mathbf{P}}
\newcommand{\NP}{\mathbf{NP}}

\newtheorem {theorem}{Theorem}
\newtheorem {observation}{Observation}
\newtheorem {inv}{Invariant}

\newtheorem {lemma}[theorem]{Lemma}

\newcommand{\myparagraph}[1]{\vskip 2mm \noindent{\bf #1\ }}

\begin{document}

\title{Nonblocker in $H$-minor free graphs:\\ kernelization meets discharging\thanks{Work supported by the National Science Centre (grant N206 567140). }}

\date{}

\author{
  \L{}ukasz Kowalik\\
  Institute of Informatics, University of Warsaw, Poland\\
  \texttt{kowalik@mimuw.edu.pl}}

\maketitle

\begin{abstract}
Perhaps the best known kernelization result is the kernel of size $335k$ for the {\sc Planar Dominating Set} problem by Alber et al.~\cite{afn:planar-domset}, later improved to $67k$ by Chen et al.~\cite{cfkx:duality-and-vertex}. This result means roughly, that the problem of finding the smallest dominating set in a planar graph is easy when the optimal solution is small. 
On the other hand, it is known that {\sc Planar Dominating Set} parameterized by $k'=|V|-k$ (also known as {\sc Planar Nonblocker}) has a kernel of size $2k'$. This means that {\sc Planar Dominating Set} is easy when the optimal solution is very large. We improve the kernel for {\sc Planar Nonblocker} to $\frac{7}{4}k'$. This also implies that {\sc Planar Dominating Set} has no kernel of size at most $(\frac{7}{3}-\epsilon)k$, for any $\epsilon>0$, unless $\P=\NP$. This improves the previous lower bound of $(2-\epsilon)k$ of~\cite{cfkx:duality-and-vertex}. Both of these results immediately generalize to $H$-minor free graphs (without changing the constants).
  
In our proof of the bound on the kernel size we use a variant of the {\em discharging method} (used e.g.\ in the proof of the four color theorem). We give some arguments that this method is natural in the context of kernelization and we hope it will be applied to get improved kernel size bounds for other problems as well.

As a by-product we show a result which might be of independent interest: every $n$-vertex graph with no isolated vertices and such that every pair of degree 1 vertices is at distance at least 5 and every pair of degree 2 vertices is at distance at least 2 has a dominating set of size at most $\frac{3}7n$.
\end{abstract}

\section{Introduction}

For many NP-complete problems there are kernelization algorithms, i.e.\ efficient algorithms which replace the input instance with an equivalent, but often much smaller one. More precisely, a {\em kernelization algorithm} takes an instance $I$ of size $n$ and a parameter $k\in\mathbb{N}$, and after time polynomial in $n$ it outputs an instance $I'$ (called a {\em kernel}) with a parameter $k'$ such that $I$ is a yes-instance iff $I'$ is a yes instance, $k'\le k$, and $|I'|\le f(k)$ for some function $f$ depending only on $k$. The most desired case is when the function $f$ is polynomial, or even linear (then we say that the problem admits a polynomial or linear kernel). In such case, when the parameter $k$ is relatively small, the input instance, possibly very large, is ``reduced'' to a small one (preferably of size polynomial, or even linear in $k$). 

\myparagraph{Kernelization and discharging.}
 A typical kernelization algorithm processes an instance of an NP-complete graph problem in the following way: roughly, as long as possible it finds a {\em reducible configuration} in the graph, i.e.\ a structure which can be replaced by a smaller structure so that the original graph is a yes-instance iff so is the new graph. Then it is shown that the kernel, i.e.\ a graph which contains no reducible configuration is small.
 
 Many  results in graph theory, including the four colour theorem as the best known example, are proven in the following way.
 Assume we are to show that graphs in some family (e.g.\ planar graphs) have some property (e.g.\ are 4-colorable).
 Then we specify a set of {\em reducible configurations}, i.e.\ structures which can be replaced by smaller structures so that the original graph has the desired property iff the new graph also has the property. Now, if a graph in our family contains such a configuration, we can proceed by induction.
 Otherwise, i.e.\ if a graph contains no reducible configuration we derive a contradiction. In the known proofs of the four color theorem~\cite{ah,rsst} (and many other results, e.g.~\cite{borodin,skreko}) this second part is realized by so-called {\em discharging method}.
 
 Since the two situations described above are so similar it is natural to ask whether the discharging method can be used to bound the size of a kernel.
 In this paper we present a result of that kind. Discharging used in the cited works for planar graphs is based on Euler's formula. Here we do not use the Euler's formula but the common theme is the same: using discharging we show that the graph under consideration cannot be ``hard everywhere'', i.e.\ even if it has some parts which are hard to dominate, then it has some parts which are easy, and on the average we get the desired bound. 
 A similar ``amortized analysis'' has been recently used in the context of kernelization by Kanj and Zhang~\cite{kanj-wg11}.

\myparagraph{Small kernels for planar graph problems.}
Perhaps the best known kernelization result is the kernel of size $335k$ for the {\sc Planar Dominating Set} problem by Alber et al.~\cite{afn:planar-domset}. This result opened a new research direction, which culminated in general results which show linear/polynomial kernels for large classes of problems in various graph families that contain planar graphs, e.g.\ bounded genus graphs or even $H$-minor free graphs~\cite{fomin:bidim-kernels,meta-kernelization}. There are several motivations for restricting the input to planar or $H$-minor free graphs. 
First, for many problems (including {\sc Dominating Set}) in general graphs no polynomial kernels exist (under appropriate complexity assumptions). 
Second, even if for some problem there is a polynomial kernel for general graphs, when executed on a planar graph it usually outputs a non-planar kernel, and then we do not want to use it because when we want to {\em solve} the kernel, we often prefer to use specialized (and faster) algorithms for planar graphs.
Finally, it is often the case that for the special case of planar graphs there is a specialized kernelization algorithm which outputs a smaller kernel than that for the general setting.
Indeed, as it was shown by Fomin et al.~\cite{fomin:bidim-kernels} many natural graph problems have a linear kernel for planar graphs. Knowing this, further research is done to reduce the leading constant in the linear function describing the kernel size. For example, the kernel of Alber et al. was later improved to $67k$ by Chen et al.~\cite{cfkx:duality-and-vertex}; the first linear kernel for {\sc Planar Connected Vertex Cover} was that of size $14k$ due to Guo and Niedermeier~\cite{guon:planarkernels} and it was then reduced to $4k$ by Wang et al.~\cite{mfcs} and even to $\frac{11}{3}k$ by Kowalik et al.~\cite{moja-cvc-11/3}. Observe that these constants may be crucial: since we deal with NP-complete problems, in order to find an exact solution in the reduced instance, most likely we need exponential time (or at least superpolynomial, because for planar graphs $2^{O(\sqrt{k})}$-time algorithms are often possible), and these constants appear in the exponents.

\myparagraph{Our Results.}
In this paper we study kernelization of the following problem restricted to planar graphs, or more generally to $H$-minor free graphs:

\defparproblemu{\nbname}{an $n$-vertex graph $G=(V,E)$ and an integer $k\in \mathbb{N}$}{$k$}{Is there a dominating set of size $n-k$?}

This problem can be also defined as {\sc Dominating Set} parameterized by $n-k$, in other words {\sc Nonblocker} is the parametric dual of {\sc Dominating Set} (see~\cite{cfkx:duality-and-vertex} for the definition of the parametric dual). {\sc Nonblocker} has a trivial $2k$-kernel for general graphs (and also for any reasonable graph class) since every $n$-vertex graph with no isolated vertices has a dominating set of size at most $n/2$ (consider a spanning forest of $G$, $2$-color it and choose the larger color class). This was improved to a $(\frac{5}{3}k+3)$-kernel by Dehne et al.~\cite{dehne-nonblocker}. Their kernelization algorithm applies the so-called catalytic rule, which identifies the neighbors of two degree 1 vertices, then removes one of the degree 1 vertices and decreases $k$ by 1 (when there is only one degree 1 vertex left, they use a classic result of McCuaig and Shepard~\cite{mccuaig} which states that any $n$-vertex graph of minimum degree 2 has a dominating set of size at most $\frac{2}5n$, for $n$ large enough). As we see the catalytic rule preserves neither planarity nor excluded minors. It follows that the best kernel for {\sc Planar Nonblocker} to date is still the trivial $2k$. A natural question arises: can this bound be improved? In this work we answer this question affirmatively: we present a $\frac{7}{4}k$-kernel for {\sc Planar Nonblocker}. Since in our reduction rules we only remove edges/vertices or contract edges our result immediately generalizes to $H$-minor-free graphs (with the same constant in the kernel size, which is a rather rare phenomenon in the field).

An important motivation for studying parametric duals, discovered by Chen et al.~\cite{cfkx:duality-and-vertex}, is that if the dual problem admits a kernel of size at most $\alpha k$, then the original problem has no kernel of size at most $(\alpha/(\alpha-1)-\epsilon)k$, for any $\epsilon>0$, unless P$=$NP. Hence, our kernel implies that {\sc Planar Dominating Set} has no kernel of size at most $(\frac{7}{3}-\epsilon)k$ for any $\epsilon>0$ (and the same holds for {\sc Dominating Set} restricted to any graph family closed under taking minors). This is the first improvement over the $(2-\epsilon)k$ lower bound of Chen et al.~\cite{cfkx:duality-and-vertex}.

We note here that although using the approach of Dehne et al.~\cite{dehne-nonblocker} one can get a $(\frac{5}{3}k+3)$-kernel for the {\em annotated version} of {\sc Planar Nonblocker} (where the instance is extended by a subset of vertices that do not need to be dominated), it is unclear how to use this result to get an improved lower bound for the kernel size of {\sc Planar Dominating Set}.

To bound the size of a kernel means just to give a lower or upper bound for the value of some graph invariant (e.g.\ the domination number) in a restricted class of graphs. Sometimes it is enough to apply a known combinatorial result, like the lower bound for the domination number of McCuaig and Shepard~\cite{mccuaig} for graphs of minimum degree 2, used by Dehne et al.~\cite{dehne-nonblocker}. There is also a better bound of $\frac{3}8n$ for  graphs of minimum degree 3 due to Reed~\cite{bruce} (later improved to $\frac{4}{11}n$ by Kostochka and Stodolsky~\cite{kostochka}). However in our kernel there still can be an unbounded number of vertices of degree 1 and 2, though there are some restrictions on them, so a tailor-made bound has to be shown. Applying the approach of Reed we show that every $n$-vertex graph with no isolated vertices and such that every pair of degree 1 vertices is at distance at least 5 and every pair of degree 2 vertices is at distance at least 2 has a dominating set of size at most $\frac{3}7n$. We suppose that this result may be of independent interest.

\myparagraph{Terminology and notation.}
We use standard graph theory terminology, see e.g.~\cite{diestel}.
By $N_G(v)$ we denote the set of neighbors of vertex $v$, and for a subset of vertices $X\subseteq V(G)$, we denote $N_G(X)=\bigcup_{x\in X}N_G(x)\setminus X$. 
The subscripts are omitted when it is clear which graph we refer to.
$G[S]$ denotes the subgraph of graph $G$ induced by a set of vertices $S$. By a $d$-vertex we mean a vertex of degree $d$. A $d$-neighbor is a neighbor of degree $d$.
We also use the Iverson bracket: $[\alpha]$ equals $1$ if the condition $\alpha$ holds and $0$ otherwise.

\section{The kernelization algorithm}

We say that a reduction rule for parameterized graph problem $P$ is {\em correct} when for every instance $(G,k)$ of $P$ it returns an instance $(G',k')$ such that:
\begin{enumerate}[a)]
\item $(G',k')$ is an instance of $P$,
\item $(G,k)$ is a yes-instance of $P$ iff $(G',k')$ is a yes-instance of $P$, and
\item $k'\le k$.
\end{enumerate}

We present six simple reduction rules below. It will be easier for us to formulate and analyze the rules for {\sc Dominating Set}. We will then convert them to rules for {\sc Nonblocker}.

\begin{enumerate}[Rule R1. ]
\item \label{r:iso} If there is an isolated vertex $v$, then remove $v$ and decrease $k$ by 1.
\item \label{r:iso-edge} If there is an isolated edge $vw$, then remove both $v$ and $w$ and decrease $k$ by 1.
\item \label{r:1-vtx} If a vertex $v$ has more than one 1-neighbors, then remove all these neighbors
except for one.
\item \label{r:adj-2} Assume there is a path $P=abcd$ with $\deg(b)=\deg(c)=2$. 
If $a\ne d$, then contract $P$ into a single vertex $v$ and decrease $k$ by one.
If $a=d$, then contract the edge $bc$.
\item \label{r:four} If there is a path $abcd$ with $\deg(a)=\deg(d)=1$, then contract edge $bc$ and decrease $k$ by one.
\item \label{r:heart} If there is a path $abcde$ with $\deg(a)=\deg(e)=1$, then remove edge $bc$.
\label{r:last}
\end{enumerate}

\begin{lemma}
\label{lem:reduce}
Rules R1-R6 are correct for {\sc Dominating Set} restricted to any minor-closed graph class.
\end{lemma}

\begin{proof}
 The condition c) is clear, while a) follows from the fact that we only remove edges or vertices and contract edges.
 It suffices to verify b).
 Let $S$ be any minimum dominating set in $G$.
 For R1 b) follows from the fact that $S$ contains $v$.
 For R2 b) follows from the fact that $S$ contains exactly one endpoint of $vw$.
 For R3 b) follows from the fact that $S$ contains $v$ and does not contain any of its 1-neighbors.
 For R4 with $a=d$ b) follows from the fact that $S$ contains $a$ and does not contain any of $\{b,c\}$.
 For R5 b) follows from the fact that $S$ contains $\{b,c\}$.
 For R6 b) follows from the fact that $S$ contains $\{b,d\}$.
 
 Finally consider R4 for $a\ne d$. This is the only nontrivial rule. Let $G'$ be the graph obtained from $G$ after applying the rule.
 First assume $(G,k)$ is a yes-instance, i.e.\ $|S|\le k$. We will show that $G'$ has a dominating set of size at most $k-1$.
 If $|\{a,b,c,d\}\cap S| \le 1$ then $\{a,b,c,d\}\cap S\varsubsetneq\{b,c\}$ for otherwise $S$ is not dominating. Assume by symmetry $b \in S$, $c\not\in S$. 
 Then $d$ has a neighbor in $S$ distinct from $c$, so $S\setminus\{b\}$ is a dominating set of size $|S|-1\le k-1$ in $G'$.
 If $|\{a,b,c,d\}\cap S| \ge 2$ then $R=S\setminus\{b,c\}\cup\{a,d\}$ is another minimum dominating set in $G$.
 Then clearly $R\setminus\{a,d\}\cup\{v\}$ is a dominating set of size $|R|-1\le k-1$ in $G'$.
 Now assume $(G',k-1)$ is a yes-instance, i.e.\ it has a dominating set $S'$ of size at most $k-1$.
 If $v\in S'$, then $S' \setminus \{v\} \cup \{a,d\}$ is a dominating set of size at most $k$ in $G$.
 If $v\not\in S'$, then $S' \cup \{b\}$ is a dominating set of size at most $k$ in $G$.
 This finishes the proof.
 
\end{proof}

Now, every Rule R$i$ is converted to Rule R$i'$ as follows. Let $(G,\ell)$ be an instance of {\sc Nonblocker}. Put $k=|V(G)|-\ell$, apply R$i$ to $(G,k)$ and get $(G',k')$. Put $\ell'=|V(G')|-k'$ and return $(G',\ell')$.

\begin{lemma}
\label{lem:reduce'}
Rules R1$'$-R6$'$ are correct for {\sc Nonblocker} restricted to any minor-closed graph class.
\end{lemma}

\begin{proof}
 The conditions a) and b) follow from Lemma~\ref{lem:reduce}. In order to prove c) we need to verify that for every $i=1,\ldots,6$ rule R$i$ does not increase $|V(G)|-k$. Indeed, this value does not change in R1, R4 when $a\ne d$, R5 and R6, decreases by 1 in R2 and R4 when $a=d$, and decreases by the number of removed vertices in R3.
\end{proof}

We note here that by the Graph Minor Theorem any minor-closed graph class can be characterized by a finite set of forbidden minors, so in particular our rules are correct for $H$-minor-free graphs.

\begin{observation}
\label{obs:rules}
 If none of the reduction rules applies to an $n$-vertex graph $G$ then $G$ has no isolated vertices, every pair of 1-vertices is at distance at least 5 and every pair of 2-vertices is at distance at least 2.
\end{observation}

The next section is devoted to the proof of the following theorem, which is the main technical contribution of this work.

\begin{theorem}
\label{thm:main}
 Every graph with no isolated vertices and such that every pair of 1-vertices is at distance at least 5 and every pair of 2-vertices is at distance at least 2 has a dominating set of size at most $\frac{3}7n$ and it can be found in polynomial time.
\end{theorem}

Let $(G,k)$ be the input instance of {\sc Nonblocker}.
Our kernelization algorithm applies rules R1-R6 as long as possible. It is clear that it can be checked in polynomial time whether a particular rule applies, and each rule is applied in linear time. Since in every rule $|V(G)|+|E(G)|$ decreases, it follows that the whole algorithm works in polynomial time (it can be even implemented in $O(n\log n)$ time but we skip the details). Let $(G',k')$ be the resulting instance. Since all the rules are correct from c) it follows that $k'\le k$. If $k' \le \frac{4}7|V(G')|$ then by Observation~\ref{obs:rules} and Theorem~\ref{thm:main} we know that $G'$ has a dominating set of size at most $\frac{3}7|V(G')|\le|V(G')|-k'$ and the algorithm returns the answer YES. Otherwise $|V(G')|\le\frac{7}4k'\le\frac{7}4k$ so we get a $\frac{7}4k$-kernel. 

\section{Proof of Theorem~\ref{thm:main}, basic setup}

In our proof we extend the approach of Reed's seminal paper~\cite{bruce}.
Let us introduce some basic notation (mostly coming from~\cite{bruce}).

Whenever it does not lead to ambiguity, if $P$ is a path then $P$ refers also to the set of vertices of $P$.
The {\em order} of a path $P$, denoted by $|P|$ is the number of its vertices (as opposed to the {\em length} of $P$ which is the number of edges, i.e.\ $|P|-1$).  
For $i\in \{0,1,2\}$, a path $P$ is an {\em $i$-path}, if $|P|\equiv i \pmod{3}$ (note that we modify the standard definition here but we prefer to be consistent with~\cite{bruce}). 
A {\em dangling path} in a graph $G$ is a path of order two with exactly one endpoint of degree 1 in $G$.

If $x$ is a vertex of a path $P$ and $P-x$ consists of an $i$-path and a $j$ -path, then $x$ is called an $(i,j)$-vertex
of $P$.
An endpoint $x$ of a path $P$ in graph $G$ is an {\em out-endpoint} if $x$ has a neighbor outside of $P$. 

A {\em vdp-cover} of a graph $G$ is a set $S$ of vertex-disjoint paths that contain all vertices of $G$.
By $S_i$ we denote the set of $i$-paths in $S$.

The idea of Reed's paper~\cite{bruce} is to find a carefully selected vdp-cover $S$ and then consider the paths of $S$ one by one and for each such path choose some of its vertices to be in the dominating set. In~\cite{bruce} it is shown that if $G$ is of minimum degree at least 3, then the dominating set is of size at most $3/8n$. Clearly, for any path $P$ of $S$ it is enough to choose $\lceil |P|/3 \rceil$ vertices to dominate the whole $P$. If $P$ is a 0-path, or if $P$ is long enough then this is at most $\frac{3}8|P|$. Hence only short 1- and 2-paths remain. A careful analysis in~\cite{bruce} shows that for each short 1-path (resp. 2-path) $P$, if $G[P]$ does not contain a dominating set of size $\lfloor |P|/3 \rfloor$ then one (resp. two) of its endpoints has a neighbor on some path different from $P$ and this neighbor can dominate the endpoint. In our case, when vertices of degree 1 and 2 are allowed this is not always possible: $G[P]$ has fewer edges and it may happen that both endpoints are of degree 1. It turns out that the most troublesome paths are the dangling paths and the paths of order 8. Our strategy is to find a cover that avoids such paths as much as possible. Although we are not able to get rid of them completely, it turns out that it is enough to exclude some configurations that contain these paths.

In the following lemma we describe the properties of the cover we use. It is an extension of the construction in~\cite{bruce}. Our contribution here is the addition of (B4)-\bref{B:last} and the explicit statement of the construction algorithm.

\begin{lemma}
\label{lem:cover}
For any graph $G$ one can find in polynomial time a vdp-cover $S$ of $G$ with the following properties.
Let $x$ be an out-endpoint of any 1-path or 2-path $P_i$ in $S$.
Let $y$ be a neighbor of $x$ on a path $P_j$, $j\ne i$ and let $P_j=P_j'yP_j''$. 
Then,
\begin{enumerate}[(B1)]
 \item \label{B:not-1} \label{B:1} $P_j$ is not a 1-path,
 \item if $P_j$ is a 0-path, then both $P_j'$ and $P_j''$ are 1-paths,
 \item if $P_j$ is a 2-path, then both $P_j'$ and $P_j''$ are 2-paths,
 \item \label{B:not-8} if $|P_j|=8$, then $P_i$ is a 1-path,
 \item \label{B:dang} if $P_j$ is a 2-path and $P_i$ is a dangling path, then either one of $P_j'$, $P_j''$ is dangling or $|P_j|\in\{11,17\}$ and $|P_j'|=|P_j''|$,
 \item \label{B:ord-5}if $P_j$ is a 2-path and $z$ is the common endpoint of $P_j$ and $P_j'$, then each neighbor of $z$ on $P_j''$ is a $(2,2)$-vertex.
 \item \label{B:0-path} every $0$-path in $S$ is of order 3.
 \label{B:last}
\end{enumerate}
\end{lemma}

\begin{proof}
 A {\em potential} of a cover $S$ is a tuple $\Phi(S)=(r_1,r_2,\ldots,r_7)$, where
\begin{itemize}
 \item $r_1 = 2|S_1|+|S_2|$,
 \item $r_2 = |S_2|$,
 \item $r_3 = \sum_{P\in S_0}|P|$,
 \item $r_4 = \sum_{P\in S_1}|P|$,
 \item $r_5$ is the number of paths of order 8 in $S$,
 \item $r_6$ is the number of dangling paths in $S$,
 \item $r_7 = n-|S_0|$.
\end{itemize}

For two covers $S$ and $S'$ with potentials $\Phi(S)=(r_1,\ldots,r_7)$ and $\Phi(S')=(r'_1,\ldots,r'_7)$ we say that $\Phi(S')<\Phi(S)$ if $\Phi(S')$ is smaller than $\Phi(S)$ in lexicographic order, i.e.\ for some $i=1,\ldots, 7$ we have $r_j=r'_j$ for $j<i$ and $r_i<r'_i$.
We will show that if one of the conditions \bref{B:1}-\bref{B:last} does not hold then we can modify the vdp-cover $S$ to get a new cover $S'$ with strictly smaller potential. It will be clear from the proof that the modification can be done in linear time. Since for every $i=1,\ldots,7$ we have $r_i=O(n)$, it follows that if we start from an arbitrary vdp-cover $S$ then after $O(n^7)$  modifications we get a cover that satisfies all of~\bref{B:1}-\bref{B:last} and the claim of the lemma will follow. 

Reed~\cite{bruce} showed that we can decrease the potential if one of \bref{B:1}-(B3) does not hold (Observations 1--3 in~\cite{bruce}, see also Lemma 1 in~\cite{kostochka}).

Assume \bref{B:not-8} does not hold, i.e.\ $|P_j|=8$ and $P_i\in S_2$. Then by (B3) both $P_j'$ and $P_j''$ are 2-paths and hence w.l.o.g.\ $|P_j'|=2$ and $|P_j''|=5$.
If $|P_i|\ne 5$, we set $S'=S\setminus\{P_i,P_j\}\cup\{P_iyP'_j, P_j''\}$. Note that both $P_iyP'_j$ and $P_j''$ are 2-paths so $r_1,\ldots,r_4$ do not change. Also $|P_j''|=5$ and $|P_iyP'_j|\ne 8$ so $r_5$ decreases.
If $|P_i|=5$ we set $S'=S\setminus\{P_i,P_j\}\cup\{P_j',P_iyP''_j\}$. Again, both $P_iyP'_j$ and $P_j''$ are 2-paths so $r_1,\ldots,r_4$ do not change. Also $|P_j'|=2$ and $|P_iyP'_j|= 11$ so $r_5$ decreases.

Assume \bref{B:dang} does not hold. Then by (B3) both $P_j'$ and $P_j''$ are 2-paths. 
By symmetry we can assume that $|P_j'|\ne 5$ since if $|P_j'|=|P_j''|=5$ then \bref{B:dang} holds.
Also, we can assume that $|P_j''|\ne 8$ since otherwise we know that $|P_j'|\ne 8$ and we can swap the names of $P_j'$ and $P_j''$ and get $|P_j'|=8\ne 5$ and $|P_j''|\ne 8$. Then we set $S'=S\setminus\{P_i,P_j\}\cup\{P_iyP'_j, P_j''\}$. Note that both $P_iyP'_j$ and $P_j''$ are 2-paths so $r_1,\ldots,r_4$ do not change. Also $|P_j''|\ne 8$ and $|P_iyP'_j|\ne 8$ so $r_5$ does not increase. Since $|P_j''|$ is not a dangling path (otherwise~\bref{B:dang} holds) and $|P_iyP'_j|\ge 5$, $r_6$ decreases by 1. 

Assume \bref{B:ord-5} does not hold. 
Let $P_j=v_1\ldots v_{3p+2}$ for some $p\ge 1$ where $v_1$ is the common endpoint of $P_j$ and $P_j'$. By (B3) $y=v_{3q}$, $1\le q\le p$.
We assumed that for some $r \ge q$ we have $v_1v_{3r+1}\in E$ or $v_1v_{3r+2}\in E$.  
If $v_1v_{3r+1}\in E$ then we consider the paths $P=v_{3p+2}v_{3p+1}\ldots v_{3r+1}v_1v_2\ldots v_{3q}P_i$ and $R=v_{3q+1}\ldots v_{3r}$ (if $q=r$ then $R$ is empty).
If $v_1v_{3r+2}\in E$ then we consider the paths $P=v_{3q+1}v_{3q+2}\ldots v_{3r+2}v_1v_2\ldots v_{3q}P_i$ and $R=v_{3r+3}\ldots v_{3p+2}$ (if $r=p$ then $R$ is empty).
We set $S'=S\setminus\{P_i,P_j\}\cup\{P,R\}$.
Note that $|R|\equiv 0\pmod 3$ and $|P|\equiv|P_i|+|P_j|\pmod 3$.
Hence if $P_i$ is a 1-path then both $P$ and $R$ are 0-paths so $r_1$ decreases and if $P_i$ is a 2-path then $P$ is a 1-path and $R$ is a 0-path, so $r_1$ stays the same and $r_2$ decreases.

If \bref{B:0-path} does not hold, we pick any $0$-path $P$ of order at least 6 and replace it by two 0-paths, one of order 3 and one of order $|P|-3$. Clearly, the potential decreases.
 \end{proof}

Let $S$ be the cover from Lemma~\ref{lem:cover}. Similarly as in~\cite{bruce}, some of the out-endpoints of the paths in $S$ will be dominated by vertices of other paths which we call accepting. Now we describe our method for finding these paths.

\noindent {\bf Accepting procedure.}
First, for every path $P\in S_1$ with at least one out-endpoint we mark exactly one, arbitrarily chosen, out-endpoint.
Second, for every path $P$ of order $|P|\in\{2,5,8\}$ and with two out-endpoints we mark both of these endpoints.

We say that vertex $v$ is a {\em neighbor} of path $P$ if $v\not\in V(P)$ and $v$ is a neighbor of a vertex of $P$.
Path $P\in S$ is {\em dangerous} if
\begin{enumerate}[$(i)$]
 \item $|P|=8$,
 \item $P$ has exactly one marked neighbor $v$,
 \item $v$ has exactly one neighbor on $P$,
 \item $\deg_G(v) > 1$,
 \item the path in $S$ that contains $v$ is of order 1,
 \item $P$ has at most one out-endpoint.
\end{enumerate}

As long as there is a non-dangerous path $P$ with a marked neighbor, we pick such a path $P$ and for its every marked neighbor $v$ we choose one vertex $w\in N(v)\cap P$ and $w$ {\em accepts} $v$. Then $v$ becomes unmarked and we call $w$ the {\em acceptor} of $v$.
If $|P|\in\{5,8\}$ and $P$ has exactly one out-endpoint $x$ we mark $x$, unless $x$ is already accepted.
This finishes the description of the accepting procedure.

All vertices that are marked after the above procedure finishes are called {\em rejected}.
A path from $S$ is {\em rejected} if it contains a rejected vertex.
The following observation follows from $(iii)$, $(iv)$ and $(v)$.

\begin{observation}
\label{obs:dangerous}
Every rejected path is of order 1 and it has at least two neighboring dangerous paths.
\end{observation}

A {\em weak path} is a path $P\in S$ such that $|P|=8$, $P$ accepts exactly one neighbor $v$, $v$ has {\em two} neighbors on $P$, the path in $S$ that contains $v$ is of order 1 and $P$ has no out-endpoints.

Consider a weak path $P=v_1\ldots v_8$. Then exactly one vertex of $P$ is an acceptor, and by (B3) it is either $v_3$ or $v_6$.
By symmetry assume $v_3$ is an acceptor. 
Then $v_3$ accepts exactly one vertex, say $v$, and $vv_6\in E$.
However, if $\deg_G(v_5)=2$ then we change the acceptor of $v$ from $v_3$ to $v_6$. 
Note that then $\deg_G(v_4)\ge 3$. 
Thus the following invariant holds.

\begin{inv}
\label{inv}
If we number vertices of a weak path $P=v_1\ldots v_8$ so that $v_3$ is an acceptor then $\deg_G(v_5)\ge 3$. 
\end{inv}

The intuition behind the notion of dangerous path is that it cannot afford accepting a vertex. As we see, a weak path is very close to being dangerous.
A weak path can afford accepting a vertex, but it needs additional help from other paths. This ``help'' is realized by the following procedure.

\noindent {\bf Forcing procedure. }
Now we define a certain set $F\subset V$. The elements of $F$ are called {\em forced vertices}.
The set $F$ is constructed by the following procedure. 
Begin with empty $F$.
Next consider weak paths of $S$, one by one.
Let $P$ be such a weak path.
If $P\cap F \ne \emptyset$ we skip $P$.
Otherwise, let us number vertices of $P=v_1\ldots v_8$ so that $v_3$ is an acceptor.
If $v_5$ has a neighbor outside $P$ then we choose exactly one such neighbor $x$, we add $x$ to $F$ and $x$ becomes {\em forced by $P$}. 
This finishes the description of the forcing procedure.

The following observation follows easily from (B3).

\begin{observation}
\label{obs:endpoint-not-forces}
If $w$ is an endpoint of a path $P\in S_1\cup S_2$ then $w\not\in F$.
\end{observation}

In what follows we construct a certain dominating set $D$. As we will see, for some paths $P$ of $S$ the ratio $|P\cap D|/|P|$ is at most $\frac{3}{7}$, and for some of them it is larger than $\frac{3}{7}$. However, we show that the later ones are amortized by the former. To this end we introduce the following discharging procedure (which is {\em not} a part of the construction of $D$ but it helps to bound $|D|$).
We assume that each vertex $v\in V$ and path $P\in S$ is assigned a rational number, called {\em charge}, which is initially 0.
By {\em sending charge} of value $\alpha$ from $x\in V\cup S$ to $y\in V\cup S$ we mean that the charge of $x$ decreases by $\alpha$ and the charge of $y$ increases by $\alpha$. The charge is sent according to the following rules.

\begin{enumerate}[{\bf Rule D1}]
 \item Let $v$ be an endpoint of a path $P\in S$ such that $v$ is accepted by a vertex $w$.
 If $P\in S_1$ and $|P|\ge 4$, then $w$ sends $\frac{4}{7}$ to $P$.
 Otherwise, i.e.\ if $|P|\in\{1,2,5,8\}$, then $w$ sends $\frac{3}{7}$ to $P$.
 \item Every rejected path sends $\frac{2}{7}$ to each neighboring dangerous path.
 \item Every dangling path sends $\frac{1}{7}$ to each neighboring path.
 \item If a vertex $x$ is forced by a weak path $P$, then $x$ sends $\frac{6}{7}$ to $P$. 
\end{enumerate}

After applying all the discharging rules above, each vertex $v$ and each path $P\in S$ ends up with some amount of charge: the total charge it received minus the total charge it sent. 
For $x\in V\cup S$ let $\ch(x)$ denote the final amount of charge at $x$.
For $P\in S$, let $\sch(P)=\ch(P) + \sum_{v\in P} \ch(v)$. 
Note that the initial total charge in $G$ is equal to 0 and it does not change by applying the discharging rules, so $\sum_{P\in S}\sch(P)=0$.

Let $A$ be the set of all acceptors. We say that a path $P\in S$ is {\em safe} when there exists a set $D_P \subseteq P$ such that
\begin{enumerate}[a)]
 \item $D_P \cup A \cup F$ dominates $P$, i.e.\ $P \subset N[D_P \cup A \cup F]$,
 \item $P\cap(A \cup F) \subseteq D_P$,
 \item $|D_P| + \sch(P) \le \frac{3}{7}|P|$.
\end{enumerate}

\begin{lemma}
 If all paths in $S$ are safe, then $G$ has a dominating set of size at most $\frac{3}7n$.
\end{lemma}

\begin{proof}
Since all paths in $S$ are safe, for each such path $P$ there is a set $D_P$ that satisfies conditions a)-c).
Then we define $D=\bigcup_{P\in S}D_P$.
By b), $A\cup F \subseteq D$. This together with a) implies that $D$ is a dominating set of $G$. 
Since $\sum_{P\in S}\sch(P)=0$, c) implies that $|D|=\sum_{P\in S} |D_P| \le \sum_{P\in S} \frac{3}{7}|P| = \frac{3}{7}n$.
 \end{proof}

In section~\ref{sec:safe} we show that $G$ satisfies the assumptions of Theorem~\ref{thm:main} then all paths in $S$ are safe. Together with the above lemma that finishes the proof of Theorem~\ref{thm:main}.

\section{All paths are safe}
\label{sec:safe} 
 
From now on we assume that $G$ satisfies the assumptions of Theorem~\ref{thm:main}.
The following lemma follows easily from the discharging rules.

\ignore{
Let $\ae$ denote the set of out-endpoints of paths in $P$ that are accepted.
\begin{lemma}
\label{lem:paths}
Let $P\in S$. 
\ignore{
If $P$ is dangerous then $\ch(P)=\frac{2}{7}$. 
If $P$ is dangling  then $\ch(P)\le-\frac{1}{7}$.
If $P$ is weak and $P$ forces a vertex then $\ch(P)=\frac{6}{7}$.
If $P$ is rejected then $\ch(P)\le-\frac{4}{7}$.
If $P$ is not any of the above kinds, }
If $P$ is neither dangerous, dangling, weak, nor rejected 
then $\ch(P)=\frac{3}{7}|P\cap\ae|$ when $|P|\in\{1,2,5,8\}$ and $\ch(P)=\frac{4}{7}|P\cap\ae|$ otherwise.
\end{lemma}
\begin{proof}
The claim follows easily from rules D2, D3 and D4. Note that Rule D2 applies at most once to $P$ because of $(ii)$, and so does Rule 4 because of the forcing procedure.
 \end{proof}
}

\begin{lemma}
\label{lem:vertices}
Let $v$ be a vertex of $G$.
If $v$ accepts a path from $S_1$ of order at least 4, then $\ch(v) \le -\frac{4}{7}-\frac{6}{7}[v\in F]$.
If $v$ accepts a path of order 1, 2, 5 or 8, then $\ch(v) \le -\frac{3}{7}-\frac{6}{7}[v\in F]$.
Otherwise $\ch(v) \le -\frac{6}{7}[v\in F]$.
 \end{lemma}

In what follows we will consider various kinds of paths in $S$ and we will show that they are safe. In many cases we will divide these paths into several subpaths, which we call ``bricks''. Then the safeness of paths from $S$ will be derived from the safeness of bricks, which is defined as follows.
We say that a path $P$ in a graph $G$ is $\alpha$-{\em safe} when there exists a set $D_P \subseteq P$ such that
\begin{enumerate}[a)]
 \item $D_P \cup A \cup F$ dominates $P$, i.e.\ $P \subset N[D_P \cup A \cup F]$,
 \item $P\cap(A \cup F) \subseteq D_P$, and
 \item $|D_P| + \sum_{v\in P}\ch(v) \le \alpha$.
\end{enumerate}
Note that if a path $P\in S$ is $\frac{3}{7}|P|$-safe it does {\em not} mean that it is safe, because in the definition of $\alpha$-safeness we ignore the charge of $P$. However the following claim is easy to verify.

\begin{lemma}
 \label{lem:safe}
 Let $P\in S$ and assume that $P=P_1\ldots P_k$.
 For $i=1,\ldots,k$ assume that path $P_i$ is $\alpha_i$-safe.
 If $\sum_{i=1}^k\alpha_i+\ch(P) \le \frac{3}{7}|P|$ then $P$ is safe. 
\end{lemma}

\begin{lemma}[3-brick Lemma]
 \label{lem:3brick}
Any path $P=v_1v_2v_3$ in $G$ such that $P\cap A \subseteq \{v_2\}$ is $\frac{8}{7}$-safe.
Moreover, if $v_2 \in A$, then $P$ is $\frac{6}7$-safe.
\end{lemma}

\begin{proof}
 By Lemma~\ref{lem:vertices}, for $v\in\{v_1,v_3\}$, $\ch(v)\le -\frac{6}{7}[v\in F]$.
 
 First assume $v_2$ is an acceptor. We put $D_P = P\cap(A \cup F)$. 
 Note that $|D_P|\le 1 + |F\cap \{v_1,v_3\}|$ and $\ch(v_2) \le -\frac{3}{7}-\frac{6}{7}[v\in F]$.
 Hence, $\sum_{v\in P}\ch(v) \le -\frac{3}{7}-\frac{6}{7}|F\cap P|$.
 It follows that $|D_P| + \sum_{v\in P}\ch(v) \le 1 + |F\cap \{v_1,v_3\}| -\frac{3}{7}-\frac{6}{7}|F\cap P| = \frac{4}{7} + \frac{1}{7}|F\cap \{v_1,v_3\}| \le \frac{6}7$.
 
 Now assume $v_2$ is not an acceptor. By Lemma~\ref{lem:vertices}, for any $v\in P$ we have $\ch(v) \le -\frac{6}{7}[v\in F]$. 
 If $|F\cap P|\ge 2$, we put $D_P = F\cap P$ and then $|D_P| + \sum_{v\in P}\ch(v) \le |F\cap P| -\frac{6}{7}|F\cap P| \le \frac{3}{7}$.
 Otherwise, i.e.\ when $|F\cap P|\le 1$, we put $D_P = \{v_2\} \cup (F\cap P)$ and then 
 $|D_P| + \sum_{v\in P}\ch(v) \le 1+|F\cap P| -\frac{6}{7}|F\cap P| = 1 + \frac{1}{7}|F\cap P| \le \frac{8}{7}$. 
 \end{proof}

\begin{lemma}[4-brick Lemma]
 \label{lem:4brick}
Any path $P=v_1v_2v_3v_4$ in $G$ such that $P\cap A \subseteq \{v_3\}$ is $2$-safe.
Moreover, if $v_3 \in A$, then $P$ is $(\frac{11}{7} + \frac{1}{7}|F\cap \{v_1,v_4\}|)$-safe (and hence $\frac{13}{7}$-safe).
\end{lemma}

\begin{proof}
 By Lemma~\ref{lem:vertices}, for $v\in\{v_1,v_2,v_4\}$, $\ch(v)\le -\frac{6}{7}[v\in F]$.
 
 First assume $v_3$ is an acceptor. We put $D_P = \{v_2\}\cup P\cap(A \cup F)$. 
 Note that $|D_P|\le 2 + |F\cap \{v_1,v_4\}|$ and $\ch(v_3) \le -\frac{3}{7}-\frac{6}{7}[v\in F]$ by Lemma~\ref{lem:vertices}.
 Hence, $\sum_{v\in P}\ch(v) \le -\frac{3}{7}-\frac{6}{7}|F\cap P|$.
 It follows that $|D_P| + \sum_{v\in P}\ch(v) \le 2 + |F\cap \{v_1,v_4\}| -\frac{3}{7}-\frac{6}{7}|F\cap P| \le \frac{11}{7} + \frac{1}{7}|F\cap \{v_1,v_4\}| \le \frac{13}{7}$.
 
 Now assume $v_3$ is not an acceptor. By Lemma~\ref{lem:vertices}, for $v\in P$ we have $\ch(v) \le -\frac{6}{7}[v\in F]$ and hence $\sum_{v\in P}\ch(v)=-\frac{6}{7}|F\cap P|$. 
 If $F\cap P=\emptyset$, we put $D_P = \{v_2,v_3\}$ and then $|D_P| + \sum_{v\in P}\ch(v) =2$.
 Finally assume $F\cap P\ne\emptyset$. If $\{v_1,v_2\}\cap F \ne \emptyset$ we put $D_P = \{v_3\} \cup (F\cap P)$ and otherwise we put $D_P = \{v_2\} \cup (F\cap P)$. Then, $|D_P| + \sum_{v\in P}\ch(v) \le 1+|F\cap P| -\frac{6}{7}|F\cap P| = 1 + \frac{1}{7}|F\cap P| \le \frac{11}{7}$. 
 \end{proof}

\begin{lemma}[7-brick Lemma]
 \label{lem:7brick}
Any path $P=v_1\ldots v_7$ in $G$ such that $P\cap A \subseteq \{v_3,v_6\}$ is $3$-safe.
\end{lemma}

\begin{proof}
 First assume $P\cap A \ne \emptyset$.
 Then we partition $P=P_1P_2$ where $P_1=v_1\ldots v_4$, $P_2=v_5v_6v_7$.
 If $P_1\cap A\ne\emptyset$ then by Lemma~\ref{lem:4brick} $P_1$ is $\frac{13}{7}$-safe and by Lemma~\ref{lem:3brick} $P_2$ is $\frac{8}{7}$-safe, so $P$ is $3$-safe.
 Otherwise, i.e.\ if $P_2\cap A\ne\emptyset$ then by Lemma~\ref{lem:4brick} $P_1$ is $2$-safe and by Lemma~\ref{lem:3brick} $P_2$ is $\frac{6}7$-safe, so $P$ is $(2+\frac{6}7)$-safe.

 Hence we can assume $P\cap A=\emptyset$.
 By Lemma~\ref{lem:vertices},  $\sum_{v\in P}\ch(v)\le -\frac{6}{7}|F\cap P|$.
 If $F\cap P=\emptyset$, we put $D_P = \{v_2,v_4,v_6\}$ and then $|D_P| + \sum_{v\in P}\ch(v) =3$.
 Finally assume $F\cap P\ne\emptyset$. 
 If $\{v_1,v_2\}\cap F \ne \emptyset$ we put $D_P = \{v_3,v_6\} \cup (F\cap P)$. 
 If $\{v_6,v_7\}\cap F \ne \emptyset$ we put $D_P = \{v_2,v_5\} \cup (F\cap P)$. 
 If $\{v_3,v_4,v_5\}\cap F \ne \emptyset$ we put $D_P = \{v_2,v_6\} \cup (F\cap P)$. 
 In every of these three cases $|D_P| + \sum_{v\in P}\ch(v) < 2+|F\cap P| -\frac{6}{7}|F\cap P| = 2 + \frac{1}{7}|F\cap P| \le 3$. 
 \end{proof}

\begin{lemma}[8-brick Lemma]
 \label{lem:8brick}
Any path $P=v_1\ldots v_8$ in $G$ such that $P\cap A =\emptyset$ is $\frac{22}7$-safe.
\end{lemma}

\begin{proof}
 By Lemma~\ref{lem:vertices},  $\sum_{v\in P}\ch(v)\le -\frac{6}{7}|F\cap P|$.
 If $|F\cap P|\le 1$, we put $D_P = \{v_2,v_5,v_8\}$ and then $|D_P| + \sum_{v\in P}\ch(v) \le 3+|F\cap P|-\frac{6}{7}|F\cap P|=3+\frac{1}{7}|F\cap P|\le\frac{22}7$.
 
 Finally assume $|F\cap P|\ge 2$. 
 If $\{v_1,v_2\}\cap F \ne \emptyset$ we put $D_P = \{v_4,v_7\} \cup (F\cap P)$. 
 If $\{v_7,v_8\}\cap F \ne \emptyset$ we put $D_P = \{v_2,v_5\} \cup (F\cap P)$. 
 Otherwise $|\{v_3,v_4,v_5,v_6\}\cap F|\ge 2$ and we can put $D_P = \{v_2,v_7\} \cup (F\cap P)$. 
 In every of these three cases $|D_P| + \sum_{v\in P}\ch(v) \le 2+|F\cap P| -\frac{6}{7}|F\cap P| = 2 + \frac{1}{7}|F\cap P| \le \frac{22}7$. 
 \end{proof}

\subsection{0-paths}

\begin{lemma}
 Every 0-path $P$ is safe.
\end{lemma}

\begin{proof}
 By (B8) $|P|=3$. 
 Clearly, $P$ may get charge only by Rule D3. Moreover, Rule D3 applies at most once to $P$ because otherwise by (B2) there are two dangling paths neighboring with the only $(1,1)$-vertex of $P$ and then there are 1-vertices at distance 4, a contradiction.
 It follows that $\ch(P)\le \frac{1}{7}$.
 By Lemma~\ref{lem:3brick} path $P$ is $\frac{8}{7}$-safe. Since $\frac{8}7 + \ch(P) \le \frac{9}7 = \frac{3}{7}|P|$ so by Lemma~\ref{lem:safe} path $P$ is safe.
 \end{proof}

\subsection{1-paths}

\begin{lemma}
 Every path $P$ of order 1 is safe.
\end{lemma}

\begin{proof}
 Let $V(P)=\{v\}$.
 Since there are no isolated vertices, $v$ is an out-endpoint.
 By (B1) $P$ does not receive charge by Rule D3.
 
 Note that $P\cap A = \emptyset$ by (B1) and $F\cap P = \emptyset$ by Observation~\ref{obs:endpoint-not-forces}.
 Hence, by Lemma~\ref{lem:vertices}, $\ch(v) = 0$.

First assume $v$ is accepted. 
 Then $P$ gets exactly $\frac{3}{7}$ by Rule D1, and Rule D2 does not apply, so $\ch(P)=\frac{3}{7}$.
 We put $D_P = \emptyset$. 
 It follows that $|D_P| + \sch(P) = \frac{3}{7} = \frac{3}{7}|P|$.
 
 If $v$ is not accepted, Rule D1 does not apply.
 Moreover, then $P$ is rejected, so by Observation~\ref{obs:dangerous} it sends $2\cdot\frac{2}{7}=\frac{4}{7}$ by Rule D2.
 Hence, $\ch(P)\le-\frac{4}7$.
 We put $D_P = \{v\}$.
 It follows that $|D_P| + \sch(P) \le 1 - \frac{4}{7} = \frac{3}{7}|P|$.
 \end{proof}

\begin{lemma}
\label{lem:1-path-endpoint}
Every 1-path $P$, $|P|\ge 4$, with an out-endpoint is safe.
\end{lemma}

\begin{proof}
 Assume $P=v_0v_1\ldots v_{3k}$ for some $k\ge 1$.
 By the accepting procedure, since $P$ has an out-endpoint, $P$ has an accepted out-endpoint and $P$ gets $\frac{4}{7}$ by D1.
 Assume w.l.o.g.\ $v_0$ is the accepted out-endpoint of $P$.
 By (B1), D3 does not apply to $P$. It follows that $\ch(P)=\frac{4}7$.
 Note that $P\cap A = \emptyset$ by (B1).
 We partition $P$ into $k+1$ paths: $P=P_0P_1\ldots P_k$, where $P_0=v_0$ and $P_i=v_{3i-2}v_{3i-1}v_{3i}$ for any $i=1,\ldots,k$.
 By Observation~\ref{obs:endpoint-not-forces} and Lemma~\ref{lem:vertices} we have $\ch(v_0)=0$, so we see that $P_0$ is $0$-safe (by choosing $D_{P_0}=\emptyset$). By Lemma~\ref{lem:3brick} for every $i=1,\ldots,k$ the path $P_i$ is $\frac{8}7$-safe.
 Since $0 + k\cdot \frac{8}7 + \ch(P) = \frac{8k+4}{7} \le \frac{9k+3}{7} = \frac{3}{7}|P|$, by Lemma~\ref{lem:safe} path $P$ is safe.  
  \end{proof}

\begin{lemma}
\label{lem:4-aux}
If a path $P\in S$ of order 4 has no out-endpoint then $P=N[x]$ for some $x\in P$. 
\end{lemma}

\begin{proof}
Let $P=v_1v_2v_3v_4$.
Since $P$ has no out-endpoint, $N(v_1)\subseteq\{v_2,v_3,v_4\}$ and $N(v_4)\subseteq\{v_1,v_2,v_3\}$.
If $v_1v_3\in E(G)$ then $N(v_3)=\{v_1,v_2,v_4\}$, so we can take $v_3$ as $x$.
Hence $v_1v_3\not\in E(G)$ and by symmetry also $v_4v_2\not\in E(G)$.
If $v_1v_4\in E(G)$ then $\deg_G(v_1)=\deg_G(v_4)=2$ and we have 2-vertices at distance 1, a contradiction.
It follows that $\deg_G(v_1)=\deg_G(v_4)=1$, and we have 1-vertices at distance 3, a contradiction.
 \end{proof}

\begin{lemma}
\label{lem:path4}
 Every path $P$ of order 4 is safe.
\end{lemma}

\begin{proof}
 Let $P=v_1v_2v_3v_4$.
 By Lemma~\ref{lem:1-path-endpoint} we can assume that $P$ has no out-endpoint.
 Then $\ch(P) = 0$, since D3 does not apply by (B1).
 Let $x\in P$ be a vertex such that $P=N[x]$, as guaranteed by Lemma~\ref{lem:4-aux}. 
 Note that $P\cap A = \emptyset$ by (B1).
 
 We put $D_P = \{x\} \cup (F\cap P)$. By Lemma~\ref{lem:vertices} and Observation~\ref{obs:endpoint-not-forces}, $\ch(v_1)=\ch(v_4)=0$ and $\ch(v_2),\ch(v_3)\le -\frac{6}{7}[v\in F]$.
 By Observation~\ref{obs:endpoint-not-forces}, $|F\cap P|\le 2$.
 Hence, $|D_P| + \sch(P) \le 1 + |F \cap P| -\frac{6}{7}|F\cap P| = 1+\frac{1}{7}|F\cap P|\le \frac{9}{7} < \frac{3}{7}|P|$.
 \end{proof}

\begin{lemma}
\label{lem:path7}
 Every 1-path $P$ of order at least 7 is safe.
\end{lemma}
\begin{proof}
 By Lemma~\ref{lem:1-path-endpoint} we can assume that $P$ has no out-endpoints. 
 By (B1), D3 does not apply to $P$. It follows that $\ch(P)=0$.

 Then we partition $P=P_0\ldots P_{(|P|-7)/3}$ where $P_0$ is of order $7$ and for every $i=1,\ldots,\frac{|P|-7}{3}$ path $P_i$ is of order 3.
 By Lemma~\ref{lem:7brick} $P_0$ is $3$-safe and by Lemma~\ref{lem:3brick} for $i=1,\ldots,\frac{|P|-7}{3}$ path $P_i$ is $\frac{8}{7}$-safe.
 Since $3+\frac{|P|-7}{3}\times \frac{8}7=\frac{8|P|+7}{21}=\frac{3}{7}|P|+\frac{7-|P|}{21}\le \frac{3}{7}|P|$ so by Lemma~\ref{lem:safe} path $P$ is safe.
 \end{proof}

\subsection{2-paths}

\ignore{
Let $Q$ be the set of neighbors of dangling paths.

\begin{lemma}
\label{lem:S_2}
Let $P \in S_2$ and let $v\in P$.
If $v$ is an acceptor then $\ch(v) \le -\frac{3}{7}+\frac{1}{7}[v\in Q]-\frac{6}{7}[v\in F]$.
If $v$ is an acceptor that accepts a 1-path of order at least 4, then $\ch(v) \le -\frac{4}{7}+\frac{1}{7}[v\in Q]-\frac{6}{7}[v\in F]$.
If $v$ is a $(2,2)$-vertex but not an acceptor then $\ch(v) \le \frac{1}{7}[v\in Q]-\frac{6}{7}[v\in F]$.
If $v$ is an endpoint then $\ch(v)=\frac{3}{7}$ if $v$ is accepted and $\ch(v)=0$ otherwise.
Otherwise i.e.\ when $v$ is an inner vertex of $P$ but not a $(2,2)$-vertex, then $\ch(v) \le -\frac{6}{7}[v\in F]$.
\end{lemma}

\begin{proof}
The claim follows easily from rules D1, D3 and D4. 
Again, Rule D3 applies at most once to $v$ because otherwise there are 1-vertices at distance 4.
From (B2), $v$ gets $\frac{1}{7}$ of charge by rule D3 only if $v$ is a $(2,2)$-vertex (and an acceptor is always a $(2,2)$-vertex). 
Endpoints are not forced by Observation~\ref{obs:endpoint-not-forces} so they do not send charge by D4.
  \end{proof}
}

\begin{lemma}
\label{lem:dangling}
 If a 2-path $P$ has a neighboring dangling path then $|P|\in\{11,17\}$ and $P$ has exactly one neighboring dangling path.
\end{lemma}

\begin{proof}
 Let $P=v_1\ldots v_{3k+2}$ and for some $i$ vertex $v_i\in P$ has a neighbor on a dangling path.
 
 Assume $|P|\not\in\{11,17\}$.
 Then by (B5) $i=3$ and $v_1v_2$ is a dangling path or $i=3k$ and $v_{3k+1}v_{3k+2}$ is a dangling path.
 In both cases there are 1-vertices at distance 4, a contradiction. 
 
 Hence  $|P|\in\{11,17\}$ and by (B5) $i=6$ when $|P|=11$ and $i=9$ when $|P|=17$.
 If $v_i$ has at least two neighbors on dangling paths then there are 1-vertices at distance 4, a contradiction. 
 \end{proof}

\begin{lemma}
 Every path $P$ of order 2 is safe.
\end{lemma}

\begin{proof}
 Let $P=v_1v_2$.
 Since there are no adjacent 1-vertices, at least one endpoint of $P$, say $v_1$ by symmetry, is an out-endpoint.
 Note that $P\cap A = \emptyset$ by (B3) and $F\cap P = \emptyset$ by Observation~\ref{obs:endpoint-not-forces}.
 Hence, by Lemma~\ref{lem:vertices}, $\ch(v_1) = \ch(v_2) = 0$.
 By Lemma~\ref{lem:dangling} $P$ does not get charge by D3.

 First assume $v_2$ is also an out-endpoint. 
 Then by (B4) $P$ does not have a neighboring path of order 8, and in particular it does not neighbor with a dangerous path, so both out-endpoints are accepted. Then we put $D_P = \emptyset$. 
 Path $P$ gets $2\times \frac{3}{7}=\frac{6}{7}$ by D1, so $\ch(P)=\frac{6}{7}$.
 It follows that $|D_P| + \sch(P) = \frac{6}{7} = \frac{3}{7}|P|$.
 
 If $v_2$ is not an out-endpoint, then $\deg_G(v_2)=1$ and $P$ is a dangling path. 
 Since $P$ has just one out-endpoint and $P$ is not accepting, D1 does not send charge to $P$.
 However, $P$ sends at least $\frac{1}{7}$ by D3 and hence $\ch(P)\le -\frac{1}7$.
 We put $D_P = \{v_1\}$.
 It follows that $|D_P| + \sch(P) = 1 - \frac{1}{7} = \frac{3}{7}|P|$.
 \end{proof}

\begin{lemma}
  Every path $P$ of order 5 is safe.
\end{lemma}

\begin{proof}
  Let $P=v_1v_2v_3v_4v_5$.
  
  By Observation~\ref{obs:endpoint-not-forces}, (B3) and Lemma~\ref{lem:vertices} we have $\ch(v_1)=\ch(v_5)=0$. 
  
  By Lemma~\ref{lem:dangling} $P$ does not get charge by D3.
  
  \case{1} Both $v_1$ and $v_5$ are out-endpoints.
  Then $P$ gets $2\times \frac{3}{7}=\frac{6}{7}$ by D1, so $\ch(P)=\frac{6}{7}$.
  We partition $P$ into three paths: $P=P_1P_2P_3$, where $P_1=v_1$, $P_2=v_2v_3v_4$ and $P_3=v_5$.
  Recall that $\ch(v_1)=\ch(v_5)=0$. 
  Moreover both $v_1$ and $v_5$ have a neighbor in $A$ so we see that $P_1$ and $P_3$  are $0$-safe (by choosing $D_{P_1}=D_{P_3}=\emptyset$). 
  By Lemma~\ref{lem:3brick} path $P_2$ is $\frac{8}7$-safe.
  Since $0 + \frac{8}7 + 0 +\ch(P) = \frac{14}{7} < \frac{3}{7}|P|$, by Lemma~\ref{lem:safe} path $P$ is safe.  

  \case{2} $P$ is accepting.
  We can assume that $P$ has at most one out-endpoint for otherwise Case 1 applies. 
  Then $P$ gets $\frac{3}{7}$ by D1, so $\ch(P)=\frac{3}{7}$. 
  Also, $A\cap P=\{v_3\}$ by (B3).
  Since 1-vertices are at distance at least 5, $v_1$ or $v_5$ is of degree at least 2, say w.l.o.g.\ $\deg(v_1)\ge 2$.
  By (B6), $v_1$ is an out-endpoint (which must be accepted by Observation~\ref{obs:dangerous}) or $v_1v_3\in E$. Hence, we put $D_P=\{v_3,v_4\}\cup(P\cap F)$ and the conditions a) and b) of the definition of safeness hold.
  Recall that $\ch(v_1)=\ch(v_5)=0$. 
  By Lemma~\ref{lem:vertices} we have $\ch(v_2),\ch(v_4)\le-\frac{6}{7}[v\in F]$ and $\ch(v_3)\le-\frac{3}{7}-\frac{6}{7}[v\in F]$. Hence, $\sum_{v\in P}\ch(v)\le-\frac{3}{7} - \frac{6}{7}|P\cap F|$ and consequently $\sch(P)\le- \frac{6}{7}|P\cap F|$.
  It follows that $|D_P| + \sch(P) \le 2 + |P\cap F \setminus\{v_3,v_4\}| - \frac{6}{7}|P\cap F|\le 2+\frac{1}{7}|P\cap F \setminus\{v_3,v_4\}| \le 2+\frac{1}7= \frac{3}7|P|$.
  
  \case{3} $P$ is non-accepting.
  We can assume that $P$ has at most one out-endpoint for otherwise Case 1 applies. 
  However, then $P$ has no out-endpoints because a 2-path can have exactly one out-endpoint only if it is accepting.
  Hence $\ch(P)=0$.
  Then we put $D_P=\{v_2,v_4\} \cup (P\cap F)$. Then $\sch(P)=\sum_{v\in P}\ch(v)\le-\frac{6}{7}|P\cap F|$.
  It follows that $|D_P| + \sch(P) \le 2 + |P\cap F \setminus\{v_2,v_4\}| - \frac{6}{7}|P\cap F|\le 2+\frac{1}{7}|P\cap F \setminus\{v_2,v_4\}| \le 2+\frac{1}7= \frac{3}7|P|$.
 \end{proof}

\begin{lemma}
\label{lem:path8}
 Every path $P$ of order 8 is safe.
\end{lemma}
\begin{proof}
Let $P=v_1\ldots v_8$.
By Lemma~\ref{lem:dangling} path $P$ does not get charge by D3.

\case{1} $P$ has both endpoints accepted.
It follows that $P$ gets $2\times \frac{3}{7}=\frac{6}{7}$ by D1 and at most $\frac{2}{7}$ by D2.
Moreover $P$ is not weak because it has out-endpoints so $P$ does not get charge by D4.
Hence, $\ch(P)\le\frac{8}{7}$.
  We partition $P$ into four paths: $P=P_1P_2P_3P_4$, where $|P_1|=|P_4|=1$, $|P_2|=|P_3|=3$.
  By Lemma~\ref{lem:vertices} $\ch(v_1)=\ch(v_8)=0$. 
  Moreover both $v_1$ and $v_8$ have a neighbor in $A$ so we see that $P_1$ and $P_4$  are $0$-safe (by choosing $D_{P_1}=D_{P_4}=\emptyset$). 
  By Lemma~\ref{lem:3brick} both paths $P_3$ and $P_4$ are $\frac{8}7$-safe.
  Since  $0 + 2\times\frac{8}7 + 0 +\ch(P) = \frac{24}{7} = \frac{3}{7}|P|$, by Lemma~\ref{lem:safe} path $P$ is safe.  

\case{2} $P\cap A = \emptyset$.
Then, according to the accepting procedure, if $P$ has an accepted endpoint then $P$ has two accepted endpoints and Case 1 applies.
Hence we can assume $P$ has no accepted endpoints and it does not get charge by D1. Since $P\cap A =\emptyset$, path $P$ is not weak so D4 does not send charge to $P$ as well.
Then $P$ can get charge only by D2 so $\ch(P)\le\frac{2}{7}$.
By Lemma~\ref{lem:8brick} $P$ is $\frac{22}7$-safe.
Since $\ch(P)+\frac{22}7\le \frac{24}7 = \frac{3}7|P|$, so by Lemma~\ref{lem:safe} path $P$ is safe.

\case{3} $|P\cap A|=2$.
By (B3), $P\cap A=\{v_3,v_6\}$.
Since $|P\cap A|>0$ we see that $P$ is not dangerous so $P$ does not get charge by D2.
Since $|P\cap A|\ne 1$ we see that $P$ is not weak so $P$ does not get charge by D4.
Let $X$ be the set of accepted endpoints of $P$. Then $P$ gets exactly $\frac{3}{7}|X|$ by D1 and hence $\ch(P)=\frac{3}{7}|X|$.
 Then we put $D_P=(\{v_2,v_3,v_6,v_7\}\setminus N_P(X)) \cup (P\cap F)$. Since $v_1,v_8\not\in F$ by Observation~\ref{obs:endpoint-not-forces}, $|D_P|=4 - |X| + |(\{v_4,v_5\}\cup N_P(X))\cap F|$.
 By Lemma~\ref{lem:vertices}, $\sum_{v\in P}\ch(v)\le 2\times (-\frac{3}{7})-\frac{6}{7}|P\cap F|=-\frac{6}{7} - \frac{6}7|P\cap F|$.
  It follows that $|D_P| + \sch(P) \le 4 - |X| + |(\{v_4,v_5\}\cup N_P(X))\cap F| + \frac{3}{7}|X| -\frac{6}{7} - \frac{6}7|P\cap F|\le \frac{22}{7}-\frac{4}7|X|+\frac{1}{7}|(\{v_4,v_5\}\cup N_P(X))\cap F|\le \frac{22}{7}-\frac{4}7|X|+\frac{1}{7}(2+|X|)= \frac{24}{7}-\frac{3}7|X|\le\frac{3}7|P|$.

\case{4} $|P\cap A|=1$ and $P$ has an accepted endpoint.
If $P$ has both endpoints accepted we apply Case 1, so we can assume $P$ has exactly one endpoint accepted, by symmetry assume it is $v_1$.
Again, since $|P\cap A|>0$ we see that $P$ is not dangerous so $P$ does not get charge by D2.
Also, $P$ does not get charge by D4 because $P$ has an accepted endpoint so $P$ is not weak.
Hence, $\ch(P)=\frac{3}{7}$.
  We partition $P$ into three paths: $P=P_1P_2P_3P_4$, where $|P_1|=1$, $|P_2|=3$ and $|P_3|=4$.
  By Lemma~\ref{lem:vertices} $\ch(v_1)=0$. Since $v_1$ has a neighbor in $A$ so we see that $P_1$ is $0$-safe (by choosing $D_{P_1}=\emptyset$). 
  By (B3), $P\cap A=\{v_3\}$ or $P\cap A=\{v_6\}$.
  In the former case path $P_2$ is $\frac{6}7$-safe by Lemma~\ref{lem:3brick} and $P_3$ are $2$-safe by Lemma~\ref{lem:4brick}.
  Since  $0 + \frac{6}7 + 2 +\ch(P) = \frac{23}{7} < \frac{3}{7}|P|$, by Lemma~\ref{lem:safe} path $P$ is safe.  
  In the case when $P\cap A=\{v_6\}$, path $P_2$ is $\frac{8}7$-safe by Lemma~\ref{lem:3brick} and $P_3$ are $\frac{13}7$-safe by Lemma~\ref{lem:4brick}.
  Since  $0 + \frac{8}7 + \frac{13}7 +\ch(P) = \frac{24}{7} = \frac{3}{7}|P|$, by Lemma~\ref{lem:safe} path $P$ is safe.  

\case{5} $|P\cap A|=1$ and $P$ has no accepted endpoints.
By symmetry we assume $P\cap A=\{v_3\}$. We consider several subcases. In each subcase we assume that the earlier subcases do not apply.
Note that $P$ can receive charge only by D4.

\case{5.1} $(P\setminus A) \cap F \ne \emptyset$.
Since $P\cap F\ne\emptyset$, by the forcing procedure $P$ does not get charge by D4, so $\ch(P)=0$.
By Observation~\ref{obs:endpoint-not-forces}, $(P\setminus A) \cap F\subseteq \{v_2,v_4,v_5,v_6,v_7\}$.
If $v_2\in F$ we put $D_P=\{v_3,v_5,v_7\} \cup (P\cap F)$.
If $v_7\in F$ we put $D_P=\{v_2,v_3,v_5\} \cup (P\cap F)$.
Otherwise, i.e.\ when $\{v_4,v_5,v_6\}\cap F \ne \emptyset$ we put $D_P=\{v_2,v_3,v_7\} \cup (P\cap F)$.
In every of these three cases $|D_P| \le 3 + |P\cap F|$.
By Lemma~\ref{lem:vertices}, $\sum_{v\in P}\ch(v) = -\frac{3}7 - \frac{6}7|P\cap F|$.
Hence, $|D_P|+\sch(P)= |D_P| + \ch(P) + \sum_{v\in P}\ch(v) = 3 + |P\cap F| +0-\frac{3}7-\frac{6}{7}|F\cap P| = \frac{18}7 + \frac{1}{7}|P\cap F| \le \frac{24}7=\frac{3}7|P|$. 
Hence in what follows we assume $(P\setminus A) \cap F = \emptyset$.

\case{5.2} At least one of the following conditions hold:
\begin{enumerate}[(C1)]
 \item $v_3$ accepts at least two vertices, or
 \item $v_3$ accepts an endpoint $v$ of a path of order at least 4, or
 \item $v_3\in F$.
\end{enumerate}
Note that by D2, $v_3$ sends at least $\frac{6}7$ in case (C1), at least $\frac{4}{7}$ in case C2 (note that $v$ is on 1-path by (B4)) and at least $\frac{9}{7}$ in case C3.
Hence $\ch(v_3)\le-\frac{4}7$. Moreover, since a weak path accepts exactly one vertex and this vertex must be on a path of order 1, and moreover a forcing weak path has no forced vertices, D4 does not apply to $P$ and $\ch(P)=0$.
We put $D_P=\{v_2,v_3,v_5,v_7\}$.
Then, $|D_P|+\sch(P)\le 4 -\frac{4}7=\frac{24}7=\frac{3}7|P|$. 

Hence, in what follows we assume $v_3$ accepts exactly one vertex, call it $v$.

\case{5.3} $\deg(v)=1$.
Then $P$ is not weak, so D4 does not apply to $P$ and $\ch(P)=0$. Moreover, since there are no 1-vertices at distance 3, $\deg(v_1)>1$.
Since $P$ has no accepted endpoints, by Observation~\ref{obs:dangerous} $P$ has no out-endpoints, and hence $v_1$ has a neighbor in $P\setminus\{v_2\}$.
If $v_3\in N(v_1)$ we put $D_P=\{v_3,v_5,v_7\}$.
If $\{v_4,v_5,v_6\}\cap N(v_1)\ne\emptyset$ we choose an $x\in \{v_4,v_5,v_6\}\cap N(v_1)$ and we put $D_P=\{v_3,x,v_7\}$.
Otherwise, i.e.\ when $\{v_7,v_8\}\cap N(v_1)\ne\emptyset$ we choose an $x\in \{v_7,v_8\}\cap N(v_1)$ and we put $D_P=\{v_3,v_5,x\}$.
In every of these three cases $|D_P| = 3$.
By Lemma~\ref{lem:vertices}, $\sum_{v\in P}\ch(v) \le \ch(v_3) \le -\frac{3}7$.
Hence, $|D_P|+\sch(P)= 3 -\frac{3}7<\frac{3}7|P|$. 

\case{5.4} $\deg(v)>1$.
We claim that $v$ has at least two neighbors on $P$.
Otherwise, $v$ has exactly one neighbor on $P$ so before accepting $v$ path $P$ satisfies the condition $(iii)$ in the definition of a dangerous path. Moreover, then $(i)$ holds as $|P|=8$, $(ii)$ holds because we excluded (C1), $(iv)$ holds because $\deg(v)>1$, $(v)$ holds because (C2) is excluded and  (B4) holds, and $(vi)$ holds because $P$ has no out-endpoints, as we observed in Case 5.3. Hence, just before accepting $v$ path $P$ was dangerous, a contradiction. 
Hence indeed $v$ has at least two neighbors on $P$. 
Then $P$ is a weak path. 
Moreover $P\cap F = \emptyset$ by Case 5.1 and (C3).
Then if $v_5v_1\in E$ we put $D_P=\{v_1,v_3,v_7\}$.
Otherwise we put $D_P=\{v_2,v_3,v_7\}$.
We claim that $D_P\cup F$ dominates $P$.
Observe that if $v_5$ has a neighbor in $\{v_1,v_2,v_3,v_7\}$ then 
in both cases $D_P$ dominates $P$. So assume the contrary.
Note that by (B6) $v_5v_8\not\in E$. 
Since by Invariant~\ref{inv} $\deg(v_5)\ge 3$ we infer that $v_5$ has a neighbor outside $P$.
Since $P\cap F = \emptyset$ we know from the procedure which builds the set of forced vertices that $v_5$ has a forced neighbor $x$ outside $P$.
It follows that $D_P\cup F$ contains $\{v_2,v_3,x,v_7\}$ and it dominates $P$ as required.
In both of our choices of $D_P$ we have $|D_P|=3$.
Since $|P\cap A| = 1$ we have $\sum_{v\in P}\ch(v) \le -\frac{3}7$.
Since D1-D3 do not send charge to $P$, we have $\ch(P)\le \frac{6}{7}$.
Hence, $|D_P|+\sch(P) \le 3 + \frac{6}{7} -\frac{3}7 = \frac{24}7 = \frac{3}7|P|$ and we see that $P$ is safe. 
 \end{proof}

\begin{lemma}
\label{lem:path11}
 Every path $P$ of order 11 is safe.
\end{lemma}
\begin{proof}
Let $P=v_1\ldots v_{11}$.
By the accepting procedure $P$ does not have accepted out-endpoints, so $P$ does not get charge by D1.
By Lemma~\ref{lem:dangling} $P$ gets at most $\frac{1}{7}$ charge by D3, so $\ch(P)\le\frac{1}{7}$.

\case{1}
$|P\cap A|\ge 2$.
Then we partition $P=P_1P_2P_3$ where $|P_1|=|P_3|=4$ and $|P_2|=3$.
Note that by (B3), for $i=1,2,3$ we have $|P_i \cap A|\le 1$.

If $P_1\cap A \ne \emptyset$ and $P_3\cap A \ne \emptyset$ then by Lemma~\ref{lem:4brick} both $P_1$ and $P_3$ are $\frac{12}{7}$-safe (note that the endpoints of $P$ are not in $F$ by Observation~\ref{obs:endpoint-not-forces}) and by Lemma~\ref{lem:3brick}  path $P_2$ is $\frac{8}{7}$-safe.
Since $\ch(P)+2\times \frac{12}7+\frac{8}7\le \frac{33}{7} = \frac{3}7|P|$, so by Lemma~\ref{lem:safe} path $P$ is safe.

By symmetry we can assume $P_1\cap A \ne \emptyset$ and $P_2\cap A \ne \emptyset$.
Then by Lemma~\ref{lem:4brick} $P_1$ is $\frac{12}{7}$-safe and $P_3$ is $2$-safe, while by Lemma~\ref{lem:3brick}  path $P_2$ is $\frac{6}{7}$-safe.
Since $\ch(P)+\frac{12}7+2+\frac{6}7\le \frac{33}{7} = \frac{3}7|P|$, so by Lemma~\ref{lem:safe} path $P$ is safe.

\case{2}
$|P\cap A| = 1$ and $P \cap F \subseteq \{v_2,v_3,v_6,v_9,v_{10}\}$.
Note that $\sum_{v\in P}\ch(v) = -\frac{3}7-\frac{6}{7}|F\cap P|\le -\frac{3}7$.
Hence $\sch(P)\le\frac{1}7-\frac{3}7=-\frac{2}7$.
Then we put $D_P=\{v_2,v_3,v_6,v_9,v_{10}\}$.
(Observe that $A\cap P \subseteq D_P$ because of (B3).)
  It follows that $|D_P| + \sch(P) \le 5 -\frac{2}{7}= \frac{33}7= \frac{3}7|P|$.

\case{3} 
$|P\cap A| = 1$ and $P \cap F \not\subseteq \{v_2,v_3,v_6,v_9,v_{10}\}$.
By Observation~\ref{obs:endpoint-not-forces} we have $\{v_4,v_5,v_7,v_8\}\cap F \ne \emptyset$.
By (B3) $P\cap A$ is equal to $\{v_3\}$, $\{v_6\}$ or $\{v_9\}$.
First assume $P\cap A=\{v_3\}$ (the case  $P\cap A=\{v_9\}$ is symmetric).
If $\{v_4,v_5\}\cap F \ne \emptyset$ we put $D_P=\{v_2,v_3,v_7,v_{10}\} \cup (P\cap F)$.
Otherwise, i.e.\ when $\{v_7,v_8\}\cap F \ne \emptyset$ we put $D_P=\{v_2,v_3,v_6,v_{10}\} \cup (P\cap F)$.
Now assume $P\cap A=\{v_6\}$. By symmetry we can assume $\{v_4,v_5\}\cap F \ne \emptyset$.
Then we put $D_P=\{v_2,v_6,v_8,v_{10}\} \cup (P\cap F)$.
In every of these three cases $|D_P| \le 4 + \min\{|P\cap F|,5\}$ (recall that $v_1,v_{11}\not\in F$ by Observation~\ref{obs:endpoint-not-forces}).
Hence, $|D_P|+\sch(P)= |D_P| + \ch(P) + \sum_{v\in P}\ch(v) < 4 + \min\{|P\cap F|,5\} +\frac{1}7-\frac{3}7-\frac{6}{7}|F\cap P| \le \frac{26}7 + \frac{1}{7}\min\{|P\cap F|,5\} \le \frac{31}7<\frac{3}7|P|$. 

\case{4}
$|P\cap A| = 0$.
Then we partition $P=P_1P_2$ where $|P_1|=8$ and $|P_2|=3$.

By Lemma~\ref{lem:8brick} $P_1$ is $\frac{22}{7}$-safe and by Lemma~\ref{lem:3brick}  path $P_2$ is $\frac{8}{7}$-safe.
Since $\ch(P)+\frac{22}7+\frac{8}7\le \frac{31}{7} < \frac{3}7|P|$, so by Lemma~\ref{lem:safe} path $P$ is safe.
 \end{proof}

\begin{lemma}
\label{lem:long-2-paths}
 Every 2-path $P$ of order at least 14 is safe.
\end{lemma}
\begin{proof}
By the accepting procedure $P$ does not have accepted out-endpoints, so $P$ does not get charge by D1.
By Lemma~\ref{lem:dangling} $P$ gets at most $\frac{1}{7}[|P|=17]$ charge by D3, so $\ch(P)\le\frac{1}{7}[|P|=17]$.

 Then we partition $P=P_0\ldots P_{(|P|-14)/3+1}$ where both $P_0$ and $P_{(|P|-14)/3+1}$ are of order $7$ and for every $i=1,\ldots,\frac{|P|-14}{3}$ path $P_i$ is of order 3.
 By Lemma~\ref{lem:7brick}  $P_0$ and $P_{(|P|-14)/3+1}$ are $3$-safe and by Lemma~\ref{lem:3brick} for $i=1,\ldots,\frac{|P|-14}{3}$ path $P_i$ is $\frac{8}{7}$-safe.
 Since $\frac{1}7[|P|=17]+2\times 3+\frac{|P|-14}{3}\times \frac{8}7=\frac{1}7[|P|=17]+\frac{8|P|+14}{21}=\frac{3}{7}|P|+\frac{14+3[|P|=17]-|P|}{21}\le \frac{3}{7}|P|$ so by Lemma~\ref{lem:safe} path $P$ is safe.
 \end{proof}

\smallskip
\noindent {\bf Acknowledgement.} I thank Micha\l\  D\c{e}bski and Marcin Mucha for helpful discussions. I am also very grateful to anonymous reviewers for careful reading and many helpful remarks.

\bibliographystyle{abbrv}

\end{document}